\documentclass{dmtcs}
\usepackage{amssymb,amsmath,rotating}
\usepackage{tabularx}
\usepackage{enumerate}
\usepackage{vaucanson-g}
\usepackage{subfigure}

\def\Z{\mathbb Z}

\def\Q{\mathbb Q}

\def\N{\mathbb N}

\def\A{\mathcal A}

\def\le{\leqslant}
\def\leq{\leqslant}
\def\ge{\geqslant}
\def\geq{\geqslant}

\newcommand{\decdot}{\raisebox{0.1ex}{\textbf{.}}}
\newcommand{\Tb}{T_{\beta}}
\newcommand{\mb}{{-\beta}}
\newcommand{\Db}{D_{\beta}}
\newcommand{\Sb}{S_{\beta}}

\newcommand{\Edb}{\mathop{\mathsf{d}_{\beta}}}
\newcommand{\Edbs}{\mathop{\mathsf{d}^{*}_{\beta}}}
\newcommand{\Edbm}{\mathop{\mathsf{d}_{-\beta}}}
\newcommand{\Edbms}{\mathop{\mathsf{d}^{*}_{-\beta}}}

\newtheorem{theorem}{Theorem}[section]
\newtheorem{lemma}[theorem]{Lemma}
\newtheorem{proposition}[theorem]{Proposition}

\newtheorem{remark}[theorem]{Remark}
\newtheorem{example}[theorem]{Example}
\newtheorem{definition}[theorem]{Definition}

\author{Christiane Frougny\addressmark{1}
 \and Anna Chiara Lai\addressmark{2}}

\begin{document}
\title{Negative bases and automata}

\address{\addressmark{1}LIAFA, CNRS UMR 7089, case 7014, 75205 Paris Cedex 13, France,
and University Paris 8\\
\addressmark{2}SBAI, Universit\`a La Sapienza, Rome, Italy}


\maketitle

\begin{abstract}
We study expansions in non-integer negative base $\mb$
introduced by Ito and Sadahiro. Using countable automata associated with $(\mb)$-expansions,
we characterize the case where the $(\mb)$-shift is a system of finite type.
We prove that, if $\beta$ is a Pisot number, then the $(\mb)$-shift is a sofic system. In that case, 
addition (and more generally normalization on any alphabet)
is realizable by a finite transducer. We then give an on-line algorithm for
the conversion from positive
base $\beta$ to negative base $\mb$.
When $\beta$ is a Pisot number, the conversion can be realized by a finite on-line
transducer.
\end{abstract}

%%%%%%%%%%%%%%%%%%%%%%%%%%%%%%%%%%%%%%%%%%%%%%%%%%%%%%%%%%%%%%%%%%%%%%%

\section{Introduction}

Expansions in integer negative base $-b$, where $b \ge 2$,
seem to have been introduced by Gr\"unwald in~\cite{Grunwald},
and rediscovered by several authors, see the historical comments given by Knuth~\cite{Knu}.
The choice of a negative base $-b$ and of the alphabet $\{0,\ldots,b-1\}$ is interesting, because it provides a signless 
representation for every number (positive or negative). In this case it is easy to distinguish the sequences representing a positive integer 
from the ones representing a negative integer: denoting $(w\decdot)_{-b}:=\sum_{i=0}^k w_k(-b)^k$ for any $w=w_k\cdots w_0$ in$\{0,\ldots,b-1\}^*$ 
with no leading $0$'s, 
we have $\N=\{(w\decdot)_{-b} \mid |w| \textnormal{ is odd}\}$. The classical monotonicity between the lexicographical ordering
 on words and the represented numerical values does not hold anymore in negative base, for instance 
$3=(111\decdot)_{-2}$, $4=(100\decdot)_{-2}$ and $111  >_{lex}  100$. 
 Nevertheless it is possible to restore such a correspondence by introducing an appropriate ordering on words, in the sequel denoted by $\prec_{alt}$, 
and called the {\em alternate order}. 

Representations in negative base also appear in some complex base number systems, for instance base $\beta=2i$ since
$\beta^2=-4$ (see \cite{Frougny99} for a study of their properties from an automata theoretic
point of view). Thus, beyond the interest in the problem in itself, the authors also wish the study of negative bases to be an useful preliminar step
 to better understanding the complex case. 

Ito and Sadahiro recently introduced expansions in non-integer negative base $\mb$
in \cite{IS}.
They have given a characterization of admissible sequences, and shown that the $(\mb)$-shift is
sofic if and only if the $(\mb)$-expansion of the number $-\frac{\beta}{\beta+1}$
is eventually periodic. 

In this paper we pursue their work. The purpose of this contribution is to show that
 many properties of the positive base (integer or not) numeration systems extend to the negative base case,
 the main difference being the sets of numbers that are representable in the two different cases.
The results could seem not surprising, but this study put into light the important role played
by the order on words: the lexicographic order for the positive bases, the alternate order
for the negative bases.
 
Very recently there have been several
contributions to the study of numbers having only positive powers
of the base in their expansion, the so-called $(\mb)$-integers, in \cite{ADMP}, \cite{MPV}, and \cite{steiner}.
 
We first establish some properties of the negative integer base $-b$, that are more or less folklore. This allows
to introduce the definitions of alternate order and of short-alternate order, that are natural to order
numbers by their $(\mb)$-expansions.

We then prove a general result which is not related to numeration systems but to the alternate order,
and which is of interest in itself.
We define a
symbolic dynamical system associated with a given infinite word $s$ satisfying some
properties with respect to
the alternate order on infinite words.
We design an infinite countable automaton recognizing it.
We then are able to characterize the case when the symbolic dynamical system is sofic
(resp. of finite type).
Using this general construction 
we can prove that the $(\mb)$-shift is a symbolic dynamical system of finite type 
if and only if the $(\mb)$-expansion of $-\frac{\beta}{\beta+1}$
is purely periodic.
We also show that the entropy of the $(\mb)$-shift is equal to $\log \beta$.

We then focus on the case where $\beta$ is a Pisot number, that is to say, an
algebraic integer greater than 1 such that the modulus of its Galois conjugates is less than 1.
The natural integers and the Golden Mean are Pisot numbers.
We extend all the results known to hold true in the Pisot case for $\beta$-expansions to the
$(\mb)$-expansions. In particular we prove that, if $\beta$ is a Pisot number, then
every number from $\Q(\beta)$ has an eventually periodic $(\mb)$-expansion, and thus that
the $(\mb)$-shift is a sofic system.

When $\beta$ is a Pisot number, it is known that addition in base $\beta$
--- and more generally normalization in base $\beta$
on an arbitrary alphabet --- is realizable by a finite transducer \cite{Frougny92}.
We show that this is still the case in base $\mb$.

The conversion from positive integer base to negative integer base is realizable by a
finite right sequential transducer. 
When $\beta$ is not an integer, we give an on-line algorithm for the conversion
from base $\beta$ to
base $\mb$, where the result is not admissible. When $\beta$ is a Pisot number, the conversion can be realized by a finite 
on-line
transducer.

A preliminary version of Sections \ref{alt} and \ref{s_real} has been presented in~\cite{FrougnyLai}.
\section{Definitions and preliminaries}

\subsection{Words and automata}

An \emph{alphabet} is a totally ordered set. In this paper the alphabets are always finite. 
A finite sequence of elements of an alphabet $A$ is called a {\em word}, and
the set of words on $A$ is the free monoid $A^*$.
The empty word is denoted by $\varepsilon$.
The set of infinite (resp. bi-infinite) words on $A$ is denoted by
$A^\mathbb{N}$ (resp.
$A^\mathbb{Z}$).
Let $v$ be a word of $A^*$, denote by $v^n$ the concatenation of $v$ to itself
$n$ times, and by $v^\omega$  the infinite concatenation $vvv\cdots$.
A word of the form $uv^\omega$ is said to be {\em eventually periodic}.
A (purely) {\em periodic} word is an eventually periodic word of the form $v^\omega$.

A finite word $v$ is a \emph{factor} of a (finite, infinite or bi-infinite)
word $x$ if there exists $u$ and $w$ such that $x=uvw$.
When $u$ is the empty word, $v$ is a \emph{prefix} of $x$.
The prefix $v$ is \emph{strict} if $v \neq x$.
When $w$ is empty, $v$ is said to be a \emph{suffix} of $x$.

We recall some definitions on automata, see \cite{Eil} and \cite{Sak} for
instance.
An {\em automaton over $A$}, $\mathcal A=(Q,A,E,I,T)$, is a directed graph
labelled by elements of $A$.
The set of vertices, traditionally called {\em states}, is denoted by $Q$,
$I \subset Q$ is the set of {\em initial} states, $T \subset Q$ is the set of
{\em terminal} states and $E \subset Q \times A \times Q$ is the set of
labelled {\em edges}.
If $(p,a,q) \in E$, we write $p \stackrel{a}{\to} q$.
The automaton is {\em finite} if $Q$ is finite.
The automaton $\mathcal A$  is {\em deterministic} if $E$ is the graph of a
(partial) function from $Q \times A$ into $Q$, and if there is a unique initial
state.
A subset $H$ of $A^*$ is said to be {\em recognizable by a finite automaton},
or {\em regular}, if
there exists a finite automaton $\mathcal A$ such that $H$ is equal to the set
of labels of paths starting in an initial state and ending in a terminal state.

Recall that two words $u$ and $v$ are said to be {\em right congruent modulo} $H$ if, for
every $w$, $uw$ is in $H$ if and only if $vw$ is in $H$. It is well known that $H$
is recognizable by a finite automaton if and only if the congruence modulo $H$ has
finite index.

Let $A$ and $A'$ be two alphabets.
A {\em transducer} is an automaton $\mathcal{T}=(Q,A^* \times A'^*,E,I,T)$
where the edges of $E$ are labelled by couples in $A^* \times A'^*$.
It is said to be {\em finite} if the set $Q$ of states and the set $E$ of edges
are finite.
If $(p,(u,v),q) \in E$, we write $p \stackrel{u| v}{\longrightarrow} q$.
The \emph{input automaton} (resp. \emph{output automaton}) of such a transducer is obtained by taking the
projection of edges on the first (resp. second) component.
A transducer is said to be {\em sequential} if its input automaton is deterministic.

An on-line transducer is a particular kind of sequential transducer.
An {\em on-line transducer} with delay $\delta$,
${\mathcal A}=(Q, A \times (A' \cup \varepsilon), E,\{q_0\})$, is a sequential 
automaton
composed of a transient part and of a synchronous part, see \cite{M}. 
The set of states is equal to $Q=Q_t \cup Q_s$, where $Q_t$ is
the set of
transient states and $Q_s$ is the set of synchronous states.
In the transient part, every path
of length $\delta$ starting in the initial state $q_0$ is of the form
$$q_0 \stackrel{x_1|\varepsilon}{\longrightarrow}q_1 \stackrel{x_2|\varepsilon}{\longrightarrow
} \cdots
\stackrel{x_{\delta}|\varepsilon}{\longrightarrow}q_{\delta}$$
where $q_0, \ldots,
q_{\delta -1}$ are in $Q_t$, $x_j$ in $A$, for $1 \le j \le \delta$, and
the only edge arriving in a state of $Q_t$ is as above.
In the synchronous part, edges are labelled by elements of $A \times A'$.
This means that the transducer starts reading words of length $\le \delta$
and outputting nothing,
and after that delay, outputs serially one digit for each input digit.
If the set of states $Q$ and the set of edges $E$ are finite, the on-line
automaton is said to be finite.

The same notions can be defined for automata and transducer processing words from
right to left : they are called {\em right} automata or transducers.

\subsection{Symbolic dynamics}
Let us recall some definitions on symbolic dynamical systems or subshifts
(see~\cite[Chapter~1]{Lot} or ~\cite{LM}).
The set $A^\Z$ is
endowed with the lexicographic order, denoted $<_{lex}$,
the product topology, and the shift $\sigma$, defined
by $\sigma((x_i)_{i \in \Z})=(x_{i+1})_{i \in \Z}$.
A set $S \subseteq A^\Z$ is a {\em symbolic dynamical system}, or {\em subshift}, if it
is shift-invariant and closed for the product topology on $A^\Z$. 
A bi-infinite word $z$ {\em avoids} a set of word $X \subset A^*$ if no factor of $z$
is in $X$. The set of all words which avoid $X$ is
denoted $S_X$. A set $S \subseteq A^\Z$ is a subshift if and only if $S$ is of the form
$S_X$ for some $X$.

The same notion
can be defined for a one-sided subshift of $A^\N$. 

Let $F(S)$ be the set of factors of elements of $S$, 
let $I(S)=A^+ \setminus F(S)$ be the set of words avoided by $S$, and let $X(S)$ be the set of elements
of $I(S)$ which have no proper factor in $I(S)$.
The subshift $S$ is
{\em sofic} if and only if $F(S)$ is
recognizable by a finite automaton, or equivalently if $X(S)$ is recognizable by a finite automaton.
The subshift $S$ is of {\em finite type} if $S=S_X$ for some finite set $X$,
or equivalently if $X(S)$ is finite.

The topological entropy of a subshift $S$ is 
\begin{equation*}
h(S)=\lim_{n \to \infty} \frac{1}{n} \log(B_n(S))
\end{equation*}
where $B_n(S)$ is the number of elements of $F(S)$ of length $n$.
When $S$ is sofic, the entropy of $S$ is equal to the logarithm of
the spectral radius of the adjacency matrix of
the finite automaton recognizing $F(S)$.

\subsection{Numeration systems}
The reader is referred to~\cite[Chapter 7]{Lot} and to ~\cite{cant} for a detailed presentation
of these topics.
Representations of real numbers in a non-integer base $\beta$
were introduced by R\'enyi~\cite{Ren} under the name of
{\em $\beta$-expansions}.
Let $x$ be a real number in the
interval $[0,1]$.
A {\em
representation in base
$\beta$} (or a $\beta$-representation) of $x$ is an
infinite word $ (x_i)_{i \ge 1}$ such that
$$x= \sum_{i \ge 1} x_i \beta^{-i}.$$

Let $\mathbf x = (x_i)_{i \ge 1}$.
The \emph{numerical value} in base $\beta$ is the
function $\pi_\beta$ defined by $\pi_\beta(\mathbf x)=\sum_{i=1}^\infty x_i\beta^{-i}$.

A particular $\beta$-representation --- called the $\beta$-{\em expansion} --- 
can be computed by the ``greedy algorithm''~:
denote by $\lfloor y \rfloor$, $\lceil y \rceil$ and $\{y\}$ the lower integer part, the upper integer part and the
fractional part of a
number $y$. Set $r_0=x$ and let
for $i \ge 1$, $x_i=\lfloor \beta r_{i-1} \rfloor$, $r_i=\{\beta r_{i-1}\}$.
Then $x= \sum_{i \ge 1} x_i \beta^{-i}$. The digits $x_i$ are elements of
the canonical alphabet $A_\beta
=\{0, \ldots,\lceil \beta \rceil -1\}$.

The $\beta$-expansion of $x \in [0,1]$ will be denoted by
$\Edb(x)=(x_i)_{i \ge 1}$. If $x>1$, there exists some $k \ge 1$ such that
$x/\beta^{k}$ belongs to $[0,1)$. If $\Edb(x/\beta^{k})=(y_i)_{i \ge 1}$
then by shifting $x=(y_1 \cdots y_{k} \decdot y_{k+1} y_{k+2}\cdots)_\beta$.

An equivalent definition is obtained by using the {\em
$\beta$-transformation} of the unit
interval which is the mapping
$$\Tb : x \mapsto \beta x- \lfloor \beta x \rfloor.$$
Then $\Edb(x)=(x_i)_{i \ge 1}$ if and only if
$x_i = \lfloor \beta \Tb^{i-1}(x) \rfloor$.

If a representation ends in infinitely many zeros, like $v0^\omega$,
the  ending zeros are omitted and the representation
is said to be {\em finite}.

In the case where the $\beta$-expansion of 1 is finite, there is
a special representation playing an important role.
Let $ \Edb(1) =(t_i)_{i \ge 1}$ and set
$\Edbs(1)=\Edb(1)$ if $\Edb(1)$ is infinite and $\Edbs(1)=
(t_1 \cdots t_{m-1} (t_m - 1))^\omega$ if $\Edb(1)=t_1 \cdots t_{m-1}t_m$
is finite.

Denote by $ \Db$
the set
of $\beta$-expansions of numbers of $ [0,1)$. It
is a shift-invariant subset of $A_\beta^\N$. The {\em $\beta$-shift} $\Sb$ is the
closure of $ \Db $ and it is a subshift of $A_\beta^\Z$.
When $\beta$ is an integer, $\Sb$ is the full $\beta$-shift $A_\beta^\Z$.

\begin{theorem}[Parry\cite{Parry}]\label{parryth}
Let $\beta>1$ be a real number.
A word $(w_i)_{i\ge 1}$ belongs to $\Db $ if and only if
for all $n \ge 1$
$$w_nw_{n+1}\cdots <_{lex} \Edbs(1).$$
A word $(w_i)_{i\in \Z}$ belongs to $\Sb$ if and only if
for all $n$
$$w_nw_{n+1}\cdots  \le_{lex} \Edbs(1).$$
\end{theorem}

The following results are well-known
(see~\cite[Chapt. 7]{Lot}).
\begin{theorem}
\begin{enumerate}
\item The $\beta$-shift is sofic if and only if $\Edb(1)$
is eventually periodic.
\item The $\beta$-shift is of finite type if and only if
$\Edb(1)$ is finite.
\end{enumerate}
\end{theorem}

It is known that the entropy of the $\beta$-shift is equal to $\log \beta$.

\bigskip

If $\beta$ is a Pisot number, then every element of $\Q(\beta) \cap [0,1]$ has an
eventually periodic $\beta$-expansion, and the $\beta$-shift $S_\beta$ is
a sofic system \cite{Bertrand,Schmidt}.

Let $C$ be an arbitrary finite alphabet of integer digits. 
The {\em normalization\index{normalization} function} in base $\beta$ on $C$
$$\nu_{\beta,C} : C^{\N} \rightarrow \A_\beta^\N$$
is the partial function which maps an infinite word $\mathbf y=(y_i)_{i \ge 1}$ over $C$,
such that $0 \le  y=\sum_{i \ge 1}y_i \beta ^{-i} \le 1$, onto the $\beta$-expansion
of $y$.
It is known \cite{Frougny92} that, when $\beta$ is a Pisot number, normalization is computable by a finite transducer 
on any alphabet $C$. Note that addition is a particular case of normalization, with
$C=\{0,\ldots,2(\lceil \beta \rceil -1)\}$.

 \section{Negative integer base}\label{int}
 
 Let $b>1$ be an integer.
 It is well known, see Knuth \cite{Knu} for instance, that every integer (positive or negative) has a unique $(-b)$-representation with digits in $A_b=\{0,1,\ldots,b-1\}$.
 Every real number (positive or negative) has a $(-b)$-representation,
 not necessarily unique, since 
 $$\big( -\dfrac{1}{b(b+1)}\big)_{-b} =\decdot 1 ((b-1)0)^\omega
 =\decdot 0(0(b-1))^\omega$$ for instance. 
The representation $\decdot 1 ((b-1)0)^\omega$ will be the admissible one.

We recall some well-known facts.
\begin{proposition}
The set of $(-b)$-representations of the positive integers is
$\{u \in \{0,1,\ldots,b-1\}^* \mid u$ does not begin with  $0$ and $|u|$ is odd$\}$.
The set of $(-b)$-representations of the negative integers is
$\{u \in \{0,1,\ldots,b-1\}^* \mid u$ does not begin with  $0$ and $|u|$ is even$\}$.
\end{proposition}

Let $A$ be a finite alphabet totally ordered, and let $\min A$ be its smallest element.
\begin{definition}
The {\em alternate order} $\prec_{alt}$ on infinite words or finite words with same length on 
$A$ is defined by:
$$u_1u_2u_3\cdots \prec_{alt} v_1v_2v_3\cdots$$ if and only if there exists $k \ge 1$ such that
$$u_i=v_i \; \;  \textrm{for}   \; \; 1 \le i <k  \; \; \; \textrm{and}  \; \; \; (-1)^k(u_k-v_k)<0.$$
\end{definition}
This order was implicitely defined in~\cite{Grunwald}.

\begin{definition}
On finite words, we define the {\em short-alternate order}, denoted $\prec_{sa}$, by:
if $u=u_{1} \cdots u_{\ell} $ and $v=v_{1} \cdots v_m$ are in $A^*$, then $u \prec_{sa} v$ if and only if 
\begin{itemize}
  \item $\ell$ and $m$ are odd, and $\ell<m$, or $\ell=m$ and $(\min A)u  \prec_{alt} (\min A)v$
  \item $\ell$ and $m$ are even, and $\ell >m$, or $\ell=m$ and $u  \prec_{alt} v$
  \item $\ell<m$ and $(\min A)^{m-\ell}u \prec_{sa} v$
   \item $\ell>m$ and $u \prec_{sa} (\min A)^{\ell - m}v$.
\end{itemize}
\end{definition}

The short-alt order is analogous to the short-lex or radix order relatively to the lexicographical order.

Denote $\langle x \rangle_{-b}$ the $(-b)$-representation of $x$.
We have the following result.
\begin{proposition}
If $x$ and $y$ are integers, 
$x<y$ if and only if $\langle x \rangle_{-b} \prec_{sa} \langle y \rangle_{-b}$.\\
If $x$ and $y$ are real numbers from the interval $[-\tfrac{b}{b+1}, \tfrac{1}{b+1})$
then $x<y$ if and only if $\langle x \rangle_{-b} \prec_{alt} \langle y \rangle_{-b}$.
\end{proposition}

\begin{example}
In base $-2$, 
$\langle 3 \rangle_{-2}=111$, $\langle 4 \rangle_{-2}=100$,
$\langle 6 \rangle_{-2}= 11010$, and $111 \prec_{sa} 100  \prec_{sa} 11010 $.
\end{example}

 \begin{proposition}\label{conv_b}
 The function that maps the $b$-representation of a positive integer
 to its $-b$-representation can be realized by a finite right sequential transducer.
 \end{proposition}
 \begin{proof}
 In Fig.~\ref{cb}, $0 \le c \le b-1$, $1 \le d \le b-1$, $0 \le e \le b-2$.
 The processing is done from right to left by 2-letter blocks. 
 A finite word $x_{k-1}\cdots x_0$ which is the $b$-expansion of $x$
 is transformed
 by the transducer into a finite word $y_k \cdots y_0$ which is the
 $(-b)$-expansion of $x$.
 It is 
 straightforward to transform this transducer into a finite
 right sequential transducer.

 \begin{figure}[h]
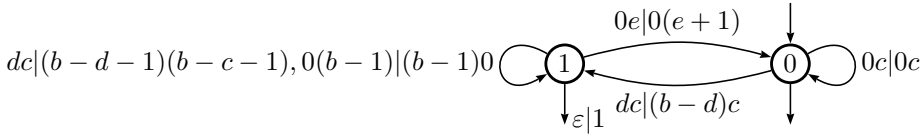

 \begin{center}
 \VCDraw{%
 \begin{VCPicture}{(-1,-1)(4,2)}
 % states
 \State[1]{(-1,0)}{A}
 \State[0]{(4,0)}{B}
 \Initial[n]{B} \Final[s]{B} \FinalL{s}{A}{\varepsilon|1}
 % transitions
 \ArcL[.5]{B}{A}{dc|(b-d)c}
 \ArcL[.5]{A}{B}{0e|0(e+1)}
 \LoopW[.5]{A}{dc|(b-d-1)(b-c-1), 0(b-1)|(b-1)0}
 \LoopE[.5]{B}{0c|0c}

 \end{VCPicture}%
         }
 \end{center}
 \caption{Finite right sequential transducer realizing conversion from base $b$ to base $-b$}\label{cb}
 \end{figure}
\end{proof}
 
 \begin{example}\label{conv2}
 Base $-2$. 
  
  \begin{figure}[h]
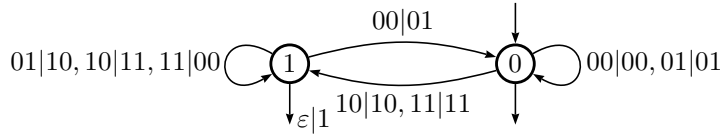

 \begin{center}
 \VCDraw{%
 \begin{VCPicture}{(-1,-1)(4,2)}
 % states
 \State[1]{(-1,0)}{A}
 \State[0]{(4,0)}{B}
 \Initial[n]{B} \Final[s]{B} \FinalL{s}{A}{\varepsilon|1}
 % transitions
 \ArcL[.5]{B}{A}{10|10,11|11}
 \ArcL[.5]{A}{B}{00|01}
 \LoopW[.5]{A}{01|10,10|11,11|00}
 \LoopE[.5]{B}{00|00,01|01}

 \end{VCPicture}%
         }
 \end{center}
 \caption{Finite right sequential transducer realizing conversion from base $2$ to base $-2$}\label{c2}
 \end{figure}
 \end{example}

\section{Symbolic dynamical systems and the alternate order}\label{alt}

We have seen in the previous section that the alternate order is the tool to compare
numbers written in a negative base. In this section we give general results on 
symbolic dynamical systems defined by the alternate order. This is analogous 
to the symbolic dynamical systems defined by the lexicographical order, see~\cite{cant}.
Let $A$ be a totally ordered finite alphabet.
\begin{definition}
A word $s=s_1s_2 \cdots$ in $A^\N$ is said to be an {\em alternately shift minimal} word (asmin-word for short)
if $s_1=\max A$ and
$s$ is smaller than, or equal to, any of its shifted images in the alternate order:
for each $n \ge 1$,
$s \preceq_{alt} s_n s_{n+1}\cdots$. 
\end{definition}

Let 
$$S(s)=\{w=(w_i)_{i  \in \Z} \in A^\Z \mid \forall
n, \; s \preceq_{alt} w_nw_{n+1} \cdots \}.$$
We construct a countable infinite automaton ${\cal A}_{S(s)}$ as follows (see Fig.~\ref{autom}, where
$[a,b]$ denotes the set $\{a,a+1,\ldots,b\}$ if $a \le b$, $\varepsilon$ else.
It is assumed in Fig.~\ref{autom} that 
$s_1 >s_j$ for $j \ge 2$.)
The set of states is $\N$.
For each state $i \ge 0$, there is an edge $i \stackrel{s_{i+1}}{\longrightarrow} i+1$.
Thus the state $i$ is the name corresponding to the path labelled $s_1 \cdots s_{i}$.
If $i$ is even, then
for each $a$ such that
$0 \le a \le s_{i+1}-1$, there is an edge $i \stackrel{a}{\longrightarrow} j$, where $j$ is such that
$s_1 \cdots s_j$ is the suffix of maximal length of $s_1 \cdots s_{i}a$. If $i$ is odd,
then for each $b$ such that $s_{i+1}+1 \le b \le s_1-1 $,
there is an edge $i \stackrel{b}{\longrightarrow} j$ where $j$ is maximal such that $s_1 \cdots s_j$ is a suffix of $s_1 \cdots s_{i}b$;
and if $s_{i+1} < s_1$ there is one edge $i \stackrel{s_1}{\longrightarrow} 1$.
By contruction, the deterministic automaton ${\cal A}_{S(s)}$ recognizes exactly the words $w$ such that every suffix of $w$ is
$\succeq_{alt} s$ and the result below follows.

\begin{figure}[h]
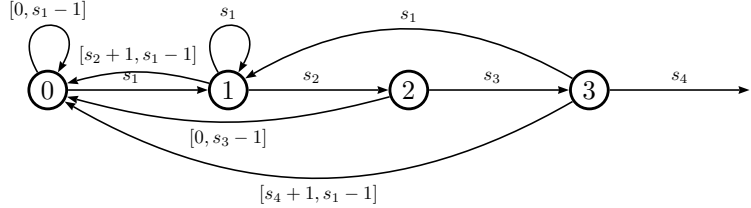

\centering
\VCDraw{%
\begin{VCPicture}{(-1,-2)(11,2.8)}
%\begin{VCPicture}{(-1,-2)(11,2.5)}
%
%\SmallState
\State[0]{(-2,0)}{A}\State[1]{(2,0)}{B}
\State[2]{(6,0)}{C}\State[3]{(10,0)}{D}
\HideState\State{(14,0)}{E}
\ChgEdgeLabelScale{0.7}
\EdgeL{A}{B}{s_1}
\EdgeL{B}{C}{s_2}
\EdgeL{C}{D}{s_3}

\LoopN[0.5]{A}{[0,s_1-1]}\LoopN[0.5]{B}{s_1}
\ArcR[0.5]{B}{A}{[s_2+1,s_1-1]}
\ArcL[0.5]{C}{A}{[0,s_3-1]}
\LArcR[0.5]{D}{B}{s_1}

\LArcL[0.5]{D}{A}{[s_4+1,s_1-1]}

\EdgeL[0.5]{D}{E}{s_4}

\end{VCPicture}
}
\caption{The automaton ${\cal A}_{S(s)}$}
\label{autom}
\end{figure}

\begin{proposition}\label{aut}
The subshift $S(s)=\{w=(w_i)_{i  \in \Z} \in A^\Z \mid \forall
n, \; s \preceq_{alt} w_nw_{n+1} \cdots \}$ is recognizable by the countable infinite automaton ${\cal A}_{S(s)}$.
\end{proposition}

\begin{proposition}\label{gensof}
The subshift $S(s)=\{w=(w_i)_{i  \in \Z} \in A^\Z \mid \forall
n, \; s \preceq_{alt} w_nw_{n+1} \cdots \}$ is sofic if and only if $s$ is eventually periodic.
\end{proposition}

\begin{proof}
The subshift $S(s)$ is sofic if and only if the set of its finite factors $F(S(s))$ is recognizable by 
a finite automaton. Given a word $u$ of $A^*$, denote by $[u]$
the right class of $u$ modulo $F(S(s))$. Then in the automaton ${\cal A}_{S(s)}$, for each state
$i \ge 1$, $i=[s_1 \cdots s_i]$, and
$0=[\varepsilon]$.
Suppose that $s$ is eventually periodic, $s=s_1 \cdots s_m(s_{m+1}  \cdots s_{m+p})^\omega$, with $m$ and $p$ minimal. Thus, for each $k \ge 0$ and each $0 \le i \le p-1$, $s_{m+pk+i}=s_{m+i}$. \\
{\sl Case 1}: $p$ is even. Then $m+i=[s_1 \cdots s_{m+i}]=[s_1 \cdots s_{m+pk+i}]$ for every 
$k \ge 0$ and $0 \le i \le p-1$. Then the set of states of ${\cal A}_{S(s)}$ is $\{0,1,\ldots,m+p-1\}$.\\
{\sl Case 2}: $p$ is odd. Then $m+i=[s_1 \cdots s_{m+i}]=[s_1 \cdots s_{m+2pk+i}]$ for every 
$k \ge 0$ and $0 \le i \le 2p-1$. The set of states of ${\cal A}_{S(s)}$ is $\{0,1,\ldots,m+2p-1\}$.\\
Conversely, suppose that $s$ is not eventually periodic. 
Then there exists an infinite sequence of indices
$i_1<i_2<\cdots $ such that the sequences $s_{i_k}s_{i_k+1}\cdots $ are all different for all $k \ge 1$.
Take any pair $(i_j,i_\ell)$, $j,\ell \ge 1$. If $i_j$ and $i_\ell$ do not have the same parity,
then $s_1 \cdots s_{i_{j}}$ and $s_1 \cdots s_{i_{\ell}}$ are not right congruent modulo $F({S(s)})$.
If $i_j$ and $i_\ell$ have the same parity, there exists $q \ge 0$ such that
$s_{i_{j}} \cdots s_{i_{j}+q-1}=s_{i_{\ell}} \cdots s_{i_{\ell}+q-1}=v$ and, for instance, 
$(-1)^{i_{j}+q}(s_{i_{j}+q} - s_{i_{\ell}+q})>0$
(with the convention that, if $q=0$ then $v=\varepsilon$).
Then
$s_1 \cdots s_{i_{j}-1}vs_{i_{j}+q} \in F({S(s)})$,
$s_1 \cdots s_{i_{\ell}-1}vs_{i_{\ell}+q} \in F({S(s)})$, but
$s_1 \cdots s_{i_{j}-1}vs_{i_{\ell}+q}$ does not belong to $F({S(s)})$. Hence
$s_1 \cdots s_{i_{j}}$ and $s_1 \cdots s_{i_{\ell}}$ are not right congruent modulo $F({S(s)})$,
so the number of right congruence classes is infinite
and $F({S(s)})$ is thus not recognizable by a finite automaton.
\end{proof}

\begin{proposition}\label{gensft}
The subshift ${S(s)}=\{w=(w_i)_{i  \in \Z} \in A^\Z \mid \forall
n, \; s \preceq_{alt} w_nw_{n+1} \cdots \}$ is a subshift of finite type if and only if $s$ is purely periodic.
\end{proposition}

\begin{proof}
Suppose that $s=(s_1 \cdots s_p)^\omega$. Consider the finite set
$X=\{s_1 \cdots s_{n-1}b \mid b \in A, \; (-1)^n(b - s_n)<0, \;1 \le  n \le p\}$.
We show that ${S(s)}={S(s)}_X$.
If $w$ is in ${S(s)}$, then $w$ avoids $X$, and conversely.
Now, suppose that ${S(s)}$ is of finite type. It is thus sofic, and by Proposition~\ref{gensof} $s$
is eventually periodic. If it is not
purely periodic,
then $s=s_1 \cdots s_m(s_{m+1}  \cdots s_{m+p})^\omega$, with $m$ and $p$ minimal, and
$s_1 \cdots s_m \neq \varepsilon$.
Let $I=\{s_1 \cdots s_{n-1}b \mid b \in A, \; (-1)^n(b - s_n)<0, \;1 \le  n \le m\}
\cup \{s_1 \cdots s_m(s_{m+1}  \cdots s_{m+p})^{2k}$ $s_{m+1}  \cdots s_{m+n-1}b
 \mid b \in A, \; k \ge 0, (-1)^{m+2kp+n}(b - s_{m+n})<0, \;1 \le  n \le 2p\}$. 
Then $I \subset A^+ \setminus F({S(s)})$.
First, suppose there exists $1 \le j \le p$ such that
$(-1)^j(s_j - s_{m+j})<0$ and $s_1 \cdots s_{j-1} = s_{m+1}
 \cdots s_{m+j-1}$.
For $k \ge 0$ fixed, let $w^{(2k)}=s_1 \cdots s_m(s_{m+1} \cdots s_{m+p})^{2k}
s_1 \cdots
s_j \in I$. We have
$s_1 \cdots s_m(s_{m+1} \cdots s_{m+p})^{2k}s_{m+1}$ $\cdots
s_{m+j-1} \in  F({S(s)})$. On the other hand,
for $n \ge 2 $,
$s_n \cdots s_m(s_{m+1} \cdots s_{m+p})^{2k}$ is greater in the alternate
order than the
prefix of $s$ of same length,
thus $s_n \cdots s_m(s_{m+1} \cdots s_{m+p})^{2k} s_1 \cdots s_j$
belongs to $F({S(s)})$.
Hence any strict factor of $w^{(2k)}$ is in $ F({S(s)})$.
Therefore for any $k \ge 0$, $w^{(2k)} \in X({S(s)})$, and $X({S(s)})$ is thus
infinite: ${S(s)}$ is not of finite type.
Now, if such a $j$ does not exist, then for every $1 \le j \le p$, $s_j=s_{m+j}$,
and $s=(s_1 \cdots s_m)^\omega$ is purely periodic.
\end{proof}

\begin{remark}\label{remgeneral}
Let $s'=s'_1s'_2 \cdots $ be a word in $A^\N$ such that $s'_1=\min A$ and, for each $n \ge 1$,
$s'_n s'_{n+1}\cdots \preceq_{alt} s'$. 
Such a word is said to be an {\em alternately shift maximal} word.
Let $S'(s')=\{w=(w_i)_{i  \in \Z} \in A^\Z \mid \forall
n, \; w_nw_{n+1} \cdots \preceq_{alt} s'\}$.
The statements in Propositions~\ref{aut}, \ref{gensof} and \ref{gensft} are also valid for the subshift $S'(s')$ (with the automaton ${\cal A}_{S'(s')}$ constructed accordingly).
\end{remark}

\section{Negative real base}\label{s_real}

\subsection{The $(-\beta)$-shift}
Ito and Sadahiro \cite{IS} introduced a greedy algorithm to represent any real number in real base $-\beta$, $\beta>1$, and with digits in $A_\mb=\{0,1,\ldots,\lfloor \beta \rfloor\}$.
Remark that, when $\beta$ is not an integer, $A_\mb=A_\beta$.

A transformation on $I_{-\beta}=\left[-\frac{\beta}{\beta+1},\frac{1}{\beta+1}\right)$
is defined as follows:
\[ T_{-\beta}(x)=-\beta x-\lfloor -\beta x+\frac{\beta}{\beta+1}\rfloor.\]

For every real number $x \in I_{-\beta}$ denote $\Edbm(x)$ the $(-\beta)$-expansion of $x$.
Then $\Edbm(x)=(x_i)_{i \ge 1}$ if and only if $x_i=\lfloor -\beta T_{-\beta}^{i-1}(x) + \frac{\beta}{\beta+1} \rfloor$,
and $x=\sum_{i \ge 1} x_i (\mb)^{-i}$. When this last equality holds, we may also write:
\begin{equation*}\label{eshift}
 x=(\decdot x_1 x_2 \cdots)_{\mb}.
\end{equation*}

We show that the alternate order $\prec_{alt}$ on $(\mb)$-expansions gives the numerical order.
\begin{proposition}
Let $x$ and $y$ be in $I_{-\beta}$. Then
$$x<y  \iff \Edbm(x) \prec_{alt}  \Edbm(y).$$
\end{proposition}
\begin{proof}
Suppose that $\Edbm(x) \prec_{alt} \Edbm(y)$. Then there exists 
$k \ge 1$ such that
$x_i=y_i $ for  $1 \le i <k $ and $(-1)^k(x_k-y_k)<0$. Suppose that $k$ is even, $k=2q$.
Then $x_{2q} \le y_{2q}-1$. Thus 
$x-y \le  -\beta^{-2q}+\sum_{i \ge 2q+1} x_i (\mb)^{-i}-\sum_{i \ge 2q+1} y_i (\mb)^{-i}<0$,
since $\sum_{i \ge 1} x_{2q+i} (\mb)^{-i}$ and $\sum_{i \ge 1} y_{2q+i} (\mb)^{-i}$ are in $I_{-\beta}$.
The case $k=2q+1$ is similar. The converse is immediate.
\end{proof}

A word $(x_i)_{i \ge 1}$ is said to be $(-\beta)$-{\em admissible} if there
exists a real number $x \in I_{-\beta}$ such that $\Edbm(x)=(x_i)_{i \ge 1}$.
The {\em $(-\beta)$-shift} $S_\mb$ is the
closure of the set of $(-\beta)$-admissible words,
and it is a subshift of $A_\beta^\Z$.

Define the sequence $\Edbms(\frac{1}{\beta+1})$ as follows:
\begin{itemize}
 \item if $\Edbm(-\frac{\beta}{\beta+1})=d_1d_2\cdots$ is not a periodic sequence with odd period,
$$\Edbms(\frac{1}{\beta+1})=\Edbm(\frac{1}{\beta+1})=0d_1d_2\cdots$$ 
\item otherwise if $\Edbm(-\frac{\beta}{\beta+1})=(d_1\cdots d_{2p+1})^\omega$,
$$\Edbms(\frac{1}{\beta+1})=(0d_1\cdots d_{2p}(d_{2p+1}-1) )^\omega.$$
\end{itemize}

\begin{theorem}[Ito-Sadahiro \cite{IS}]\label{adm}
A word $(w_i)_{i \ge 1}$ is $(-\beta)$-admissible if and only if for each $n \ge 1$
\begin{equation*}
\Edbm(-\frac{\beta}{\beta+1})\preceq_{alt} w_nw_{n+1 }\cdots \prec_{alt}
\Edbms(\frac{1}{\beta+1}).
\end{equation*}
A word $(w_i)_{i \in \Z}$ is an element of the $(-\beta)$-shift
if and only if for each $n$
\begin{equation*}
\Edbm(-\frac{\beta}{\beta+1})\preceq_{alt} w_nw_{n+1 }\cdots \preceq_{alt}
\Edbms(\frac{1}{\beta+1}).
\end{equation*}
\end{theorem}

Put $\mathbf{d}=\Edbm(-\frac{\beta}{\beta+1})=d_1d_2\cdots$ and $\mathbf{d}^*=\Edbms(\frac{1}{\beta+1})$.
Theorem~\ref{adm} can be restated as follows.

\begin{lemma}\label{inter}
If $\mathbf{d}=\Edbm(-\frac{\beta}{\beta+1})$ is not a periodic sequence with odd period, 
then
$$S_\mb=S(\mathbf{d})=\{(w_i)_{i \in \Z} \in A_\beta^\Z \mid \forall n,\; \mathbf{d} \preceq_{alt} w_nw_{n+1 }\cdots  \}.$$
If $\mathbf{d}=\Edbm(-\frac{\beta}{\beta+1})$ is a periodic sequence of odd period,
then $\mathbf{d}^*=(0d_1\cdots d_{2p}(d_{2p+1}-1) )^\omega$ and
$$S_\mb= S(\mathbf{d}) \cap S'(\mathbf{d}^*)$$
where 
$$S'(\mathbf{d}^*)=\{(w_i)_{i \in \Z} \in A_\beta^\Z \mid \forall n,\; 
 w_nw_{n+1 }\cdots \preceq_{alt} \mathbf{d}^*\}.$$
\end{lemma}

\begin{theorem}
The $(-\beta)$-shift is a system of finite type if and only if 
$\Edbm(-\frac{\beta}{\beta+1})$ is purely periodic.
\end{theorem}

\begin{proof}
If $\Edbm(-\frac{\beta}{\beta+1})$ is purely periodic with an even period, the result follows from Theorem~\ref{adm}, Lemma~\ref{inter} and Proposition~\ref{gensft}.
If $\Edbm(-\frac{\beta}{\beta+1})$ is purely periodic with an odd period,
the result follows from Theorem~\ref{adm}, Lemma~\ref{inter}, Proposition~\ref{gensft},
Remark~\ref{remgeneral}, and the fact that the intersection of two finite sets is finite.
\end{proof}

By Theorem~\ref{adm}, Lemma~\ref{inter}, Proposition~\ref{gensof},
Remark~\ref{remgeneral}, and the fact that the intersection of two regular sets is again regular the following result follows.

\begin{theorem}[Ito-Sadahiro \cite{IS}]
The $(-\beta)$-shift is a sofic system if and only if 
$\Edbm(-\frac{\beta}{\beta+1})$ is eventually periodic.
\end{theorem}

\begin{example}\label{ex1}
Let $G=\frac{1+\sqrt{5}}{2}$; then  $\mathop{\mathsf{d}_{G}}(1)=11$ and the $G$-shift is of finite type.
Since $\mathop{\mathsf{d}_{-G}}(-\frac{G}{G+1})=10^\omega$ the $(-G)$-shift
is a sofic system which is not of finite type. \\
The automaton in Fig.~\ref{shiftG} (right) recognizing the $(-G)$-shift is obtained by minimizing the result of the
construction of Proposition~\ref{aut}. Remark that it is the automaton
which recognizes the celebrated even shift (see~\cite{LM}).

\begin{figure}[h]
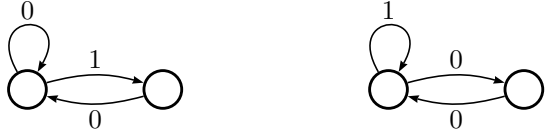

\begin{center}
\VCDraw{%
\begin{VCPicture}{(-1,-0.2)(4,1.8)}
% states
\State{(-4,0)}{A}
\State{(-1,0)}{B}

% transitions
\ArcL[.5]{A}{B}{1}
\ArcL[.5]{B}{A}{0}
\LoopN[.5]{A}{0}

% states
\State{(4,0)}{A1}
\State{(7,0)}{B1}
% transitions
\ArcL[.5]{A1}{B1}{0}
\ArcL[.5]{B1}{A1}{0}
\LoopN[.5]{A1}{1}
\end{VCPicture}
        }
\end{center}
\caption{Finite automata for the $G$-shift (left) and for the $(-G)$-shift (right)}
\label{shiftG}
\end{figure}
\end{example}
%\vspace*{-0.2cm}

\begin{example}\label{ex2}
Let $\beta=G^2=\frac{3+\sqrt{5}}{2}$; then  $\mathop{\mathsf{d}_{\beta}}(1)=21^\omega$
and the $\beta$-shift is sofic, but not of finite type. Now, $\Edbm(-\frac{\beta}{\beta+1})=(21)^\omega$
and the $(\mb)$-shift is of finite type: the set of minimal forbidden factors is $X(S_\mb)=\{20\}$.

\medskip

\begin{figure}[h]
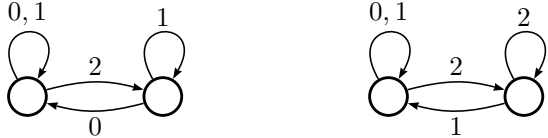

\begin{center}
\VCDraw{%
\begin{VCPicture}{(-1,-0.2)(4,1.8)}
% states
\State{(-4,0)}{A}
\State{(-1,0)}{B}

% transitions
\ArcL[.5]{A}{B}{2}
\ArcL[.5]{B}{A}{0}
\LoopN[.5]{A}{0,1}
\LoopN[.5]{B}{1}

% states
\State{(4,0)}{A1}
\State{(7,0)}{B1}
% transitions
\ArcL[.5]{A1}{B1}{2}
\ArcL[.5]{B1}{A1}{1}
\LoopN[.5]{A1}{0,1}
\LoopN[.5]{B1}{2}
\end{VCPicture}
        }
\end{center}
\caption{Finite automata for the $G^2$-shift (left) and for the $(-G^2)$-shift (right)}
\label{autmG}
\end{figure}
\end{example}

\subsection{Entropy of the $-\beta$-shift}

Examples \ref{ex1} and \ref{ex2} suggest that the entropy of the $\mb$-shift is the same as the entropy of the $\beta$-shift because the adjacency matrices of the automata are the same. This is what we show in this section.

A standard technique for computing the entropy of a subshift $S$ is to construct a (not necessarily finite) automaton recognizing $F(S)$. Then the submatrices of the adjacency matrix are taken into account and for every $n$ the greatest eigenvalue $\lambda_n$ of the submatrix of order $n$ is computed. A result proved in \cite{Ho1} ensures that the limit $\lambda$ of the sequence $\lambda_n$ exists and it satisfies $h(S)=\log \lambda$. Unfortunately the explicit computation of the $\lambda_n$'s in the general case
turns out to be very complicated, so we use tools from the theory of dynamical systems:
\begin{itemize}
 \item[--] the notion of topological entropy for one-dimensional dynamical systems, a one-dimensional dynamical system being a couple $(I,T)$ consisting in a bounded interval $I$ and a piecewise continuous transformation $T:I\rightarrow I$;
\item[--] a result by Takahashi~\cite{Tak80} establishing the relation between topological entropies of one-dimensional dynamical systems and symbolic dynamical systems;
\item[--] a result by Shultz~\cite{Shu07} on the topological entropy of some one-dimensional dynamical systems.
\end{itemize}
Let us begin with the definition of topological entropy for one-dimensional dynamical systems.

\vskip0.2cm

\begin{definition}
%[topological entropy for one-dimensional dynamical systems]
 Let $(I,T)$ be a dynamical system.

For every finite cover of $I$, say $\mathcal C$, set:
\begin{equation*}
 H(T,\mathcal C):=\lim\sup\frac{1}{n}\log N\left(\bigvee_{m=0}^{n-1}T^{-m}\mathcal C\right)
\end{equation*}
with $\bigvee$ denoting the finest common refinement and $N(\mathcal C)$ denoting the number of elements of the smallest subcover of $\mathcal C$, a subcover of $\mathcal C$ being a subfamily of $\mathcal C$ still covering $I$. 

The \emph{topological entropy} of $(I,T)$ is given by the formula
\begin{equation}
 h(I,T):=\sup H(T,\mathcal C).
\end{equation}
\end{definition}

In \cite{Tak80} Takahashi proved the equality between the topological entropy of a piecewise continuous dynamical system and the topological entropy of an appropriate subshift.
Before stating such a result we need a definition.

\begin{definition}
%[lap intervals]
Let $T:I\rightarrow I$ be a piecewise continuous map on the interval $I$. The \emph{lap intervals} $I_0,\dots,I_l$ of $T$ are closed intervals satisfying the following conditions:
\begin{enumerate}[(a)]
 \item $I_0\cup\dots\cup I_l=I$;
 \item $T$ is monotone on each interval $I_i$, $~i=0,\dots,l$;
\item the number $l$ is minimal under the conditions (a) and (b). 
\end{enumerate}
The number $l$ is called \emph{lap number} and it is denoted $lap(T)$. 
\end{definition}

\begin{remark}
 If the map $T$ is piecewise linear then the lap intervals are unique and they coincide with the intervals of continuity of $T$.
\end{remark}

\begin{theorem}[Takahashi~\cite{Tak80}] \label{t entr tak}
Let $T$ be a piecewise continuous transformation over the closed interval $I$ on itself. Let $\gamma_T:I\rightarrow A_T^\N$ be the map defined by the relation $x\mapsto x_1x_2\cdots$ with $x_n$ satisfying $T^n(x)\in I_{x_n}$. Define the subshift $X_T:=\overline{\gamma_T(I)}$ in $A^\N$. 

If $lap(T)$ is finite then:
\begin{equation}
 h(X_T)=h(I,T).
\end{equation}
\end{theorem}

The entropy in the very particular case of a piecewise linear map with constant slope is explicitely given in the following result.

\begin{proposition}[{Shultz~\cite[Proposition 3.7]{Shu07}}]\label{p entr piecew} 
 Let $T$ be a piecewise linear map with slope $\pm \beta$. Then the topological entropy of $T$ is equal to $\log \beta$.
\end{proposition}

We now prove our result.

\begin{theorem}
 The topological entropy of $S_{-\beta}$ is equal to $\log \beta$.
\end{theorem}

\begin{proof}
 Consider the dynamical system $(I_{-\beta},T_{-\beta})$. We extend the map $T_{-\beta}$ to the closure of $I_{-\beta}$ to fullfill the conditions of Theorem \ref{t entr tak}. By definition of the $(-\beta)$-expansion, the subshift $X_{T_\beta}$ coincides with the closure of the set of the $(-\beta)$-expansions in $A_\mb^\N$, whose entropy is the same as $S_{-\beta}\subset A_\mb^\Z$. As $T_{-\beta}$ is piecewise linear, the lap intervals coincide with the (finite) number of continuity intervals. Then, by Theorem \ref{t entr tak} and by Proposition \ref{p entr piecew}, $h(S_{-\beta})=h(I_{-\beta},T_{-\beta})=\log \beta$. 
\end{proof}
\subsection{The Pisot case}

We first prove that the classical result saying that if $\beta$ is a Pisot number, then every element of $\Q(\beta) \cap [0,1]$ has an eventually periodic
$\beta$-expansion is still valid for the base $\mb$. 

\begin{theorem}\label{rat}
If $\beta$ is a Pisot number, then every element of $\Q(\beta) \cap I_{-\beta}$ has an eventually periodic
$(-\beta)$-expansion.
\end{theorem}

\begin{proof}
 Let $M_\beta(X)=X^d -a_1X^{d-1}-\cdots-a_d$ be the minimal polynomial of $\beta$ and denote by $\beta=\beta_1,\ldots,\beta_d$ the conjugates of $\beta$. Let $x$ be arbitrarily fixed in $\Q(\beta) \cap I_{-\beta}$. Since $\Q(\beta) =
 \Q(-\beta)$, $x$ can be expressed as $x=q^{-1}\sum_{i=0}^{d-1} p_i (-\beta)^i$
with $q$ and $p_i$ in $\Z$, $q>0$ as small as possible in order to have uniqueness.

Let $(x_i)_{i\geq 1}$ be the $(-\beta)$-expansion of $x$, and write
\[
 r_n=r_n^{(1)}=r_n^{(1)}(x)=\frac{x_{n+1}}{-\beta}+\frac{x_{n+2}}{(-\beta)^2}+\cdots=(-\beta)^n\left(x-\sum_{k=1}^{n} x_k (-\beta)^{-k}\right).
\]
Since $r_n=T^n_{-\beta}(x)$ belongs to $I_{-\beta}$
then $|r_n|\leq \frac{\beta}{\beta+1}<1$.
For $2\leq j \leq d$, let
\[r_n^{(j)}=r_n^{(j)}(x)=(-\beta_j)^n\left(q^{-1}\sum_{i=0}^{d-1} p_i (-\beta_j)^i-\sum_{k=1}^{n} x_k (-\beta_j)^{-k}\right).
\]
Let $\eta=\max\{|\beta_j| \mid 2\leq j \leq d\}$: since $\beta$ is a Pisot number, $\eta<1$. Since $x_k \leq \lfloor \beta \rfloor$ we get
\[
 |r_n^{(j)}|\leq q^{-1}\sum_{i=0}^{d-1}| p_i | \eta^{n+i} + \lfloor \beta \rfloor \sum_{k=0}^{n-1} \eta^{k} 
\]
and since $\eta<1$, $\max_{1\leq j \leq d}\{\sup_n\{|r_n^{(j)}|\}\}< \infty.$

We need a technical result. Set $R_n=(r_n^{(1)}, \ldots, r_n^{(d)})$ and let $B$ the matrix $B=((-\beta_j)^{-i})_{1\leq i,j \leq d}.$
\begin{lemma}
 Let $x=q^{-1}\sum_{i=0}^{d-1} p_i (-\beta)^i$. For every $ n \geq 0$ there exists a unique $d$-uple $Z_n=(z_n^{(1)},\dots,z_n^{(d)})$ in $\Z^d$ such that $R_n=q^{-1}Z_nB$.
\end{lemma}

\begin{proof}
By induction on $n$. First, $r_1=-\beta x - x_1$, thus
\[
 r_1=q^{-1}\left(\sum_{i=0}^{d-1} p_i (-\beta)^{i+1}-qx_1\right)=q^{-1}\left(
\frac{z_1^{(1)}}{-\beta}+\cdots+\frac{z_1^{(d)}}{(-\beta)^d}\right)
\]
using the fact that $(-\beta)^d= -a_1(-\beta)^{d-1}+a_2 (-\beta)^{d-2}+
\cdots+(-1)^{d}a_d$. Now, $r_{n+1}=-\beta r_{n} - x_{n+1}$, hence
\[
 r_{n+1}=q^{-1}\left(z_n^{(1)}+
\frac{z_1^{(2)}}{-\beta}+\cdots+\frac{z_n^{(d)}}{(-\beta)^{d-1}}-qx_{n+1}\right)=q^{-1}\left(
\frac{z_{n+1}^{(1)}}{-\beta}+\cdots+\frac{z_{n+1}^{(d)}}{(-\beta)^{d}}\right)
\]
since $z_n^{(1)}-qx_{n+1}\in \Z$.
Thus for every $n$ there exists $(z_n^{(1)},
\ldots,z_n^{(d)})$ in $\Z^d$ such that
\[r_{n}=q^{-1}\sum_{k=1}^{d}z_n^{(k)}(-\beta)^{-k}.
\]
Since the latter equation has integral coefficients and is satisfied by $-\beta$, it is also satisfied by $-\beta_j$, $2 \leq j \leq d$, and
\[r_{n}^{(j)}=(-\beta_j)^n\left(q^{-1}\sum_{i=0}^{d-1}\bar p_i (-\beta_j)^i-\sum_{k=1}^{n} x_k (-\beta_j)^{-k}\right)=q^{-1}\sum_{k=1}^{d}z_n^{(k)}(-\beta_j)^{-k}.
\]
\end{proof}

Let us go back to the proof of Theorem~\ref{rat}. Let $V_n=qR_n$. The $(V_n)_{n\geq1}$ have bounded norm, since
$\max_{1\leq j \leq d}\{\sup_n\{|r_n^{(j)}|\}\}< \infty$. As the matrix $B$ is invertible, for every $ n\geq1$,
\[\|Z_n\|=\|(z_n^{(1)},\dots,z_n^{(d)})\|=\max\{|z_n^{(j)}|\; : 1\leq j \leq d\}<\infty\]
so there exist $p$ and $m \geq 1$ such that $Z_{m+p}=Z_{p}$, hence $r_{m+p}=r_{p}$ and the $(-\beta)$-expansion of $x$ is eventually periodic.
\end{proof}

As a corollary we get the following result.
\begin{theorem}\label{pis-sof}
If $\beta$ is a Pisot number then the $(-\beta)$-shift is a sofic system.
\end{theorem}

The \emph{normalization} in base $-\beta$
is the function which maps any $(-\beta)$-represen\-tation on an alphabet $C$ of digits
of
a given number of $I_{\mb}$ onto the admissible $(-\beta)$-expansion
of that number.

Let $C=\{-c, \ldots,c\}$, where $c \ge \lfloor \beta \rfloor$ is an integer.
Denote
$$Z_{\mb}(2c)= \Big\{(z_i)_{i \ge 0} \in\{-2c,\ldots,2c\}^\N\ \Big|\ 
\sum_{i \ge 0}z_i(\mb)^{-i}=0\Big\}\,.
$$

The set $Z_\mb(2c)$ is recognized by a countable infinite automaton $\mathcal A_\mb(2c)$:
the set of states $Q(2c)$ consists of  all
$s\in\mathbb Z[\beta] \cap [-\frac{2c}{\beta-1},\frac{2c}{\beta-1}]$.
Transitions are of the form $s\stackrel e \to s'$
with $e \in\{-c,\ldots,c\}$ such that $s'=-\beta s+e$.
The state $0$ is initial; every state is terminal.

Let $M_\beta(X)$ be the minimal polynomial of $\beta$,
and denote by $\beta=\beta_1$, $\beta_2$, \ldots, $\beta_d$ the roots of $M_\beta$.
We define a norm on the discrete lattice of rank $d$, $\Z[X]/(M_\beta)$, as
$$||P(X)||=\max_{1 \le i \le d} |P(\beta_i)|.$$

\begin{proposition}
If $\beta$ is a Pisot number then the automaton $\mathcal A_\mb(2c)$ is finite for 
every $c \ge \lfloor \beta \rfloor$.
\end{proposition}
\begin{proof}
Every state $s$ in $Q(2c)$ is
associated with the label of the shortest path
$f_0f_1\cdots f_n$ from $0$ to $s$ in the automaton. Thus $s=f_{0}(-\beta)^{n} +f_1(-\beta)^{n-1} +\cdots
 + f_n=P(\beta)$, 
with $P(X)$ in $\Z[X]/(M_\beta)$. Since $f_0f_1\cdots f_n$ is a prefix of a word of
$Z_\mb(2c)$, there exists $f_{n+1}f_{n+2}\cdots$ such that $(f_i)_{i \ge 0}$
is in $Z_\mb(2c)$. Thus
$s=|P(\beta)|<\frac {2c}{\beta-1}$.
For every conjugate $\beta_i$, $2 \le i \le d$, $|\beta_i|<1$, and
$|P(\beta_i)| < \frac {2c}{1-|\beta_i|}$.
Thus every state of $Q(2c)$ is bounded in norm, and so there is only
 a finite number of them.
\end{proof}

The {\em redundancy transducer} $\mathcal R_\mb(c)$ is similar
to $\mathcal A_\mb(2c)$.
Each transition $s\stackrel e\to s'$ of $\mathcal A_\mb(2c)$ is replaced in
$\mathcal R_\mb(c)$ by a set of transitions
$s\stackrel{a|b}\longrightarrow s'$, with $a,b\in\{-c,\ldots,c\}$ and $a-b=e$.
Thus one obtains the following proposition.

\begin{proposition}\label{redundancy}
The redundancy transducer $\mathcal R_\mb(c)$ recognizes the set
$$
\big\{(x_1x_2\cdots,y_1y_2\cdots) \in C^\N \times C^\N\ \big|\ \;
\sum_{i \ge 1}x_i(\mb)^{-i}=\sum_{i \ge 1}y_i(\mb)^{-i} \big\}.
$$
If $\beta$ is a Pisot number, then $\mathcal R_\mb(c)$ is finite.
\end{proposition}

\begin{theorem}
If $\beta$ is a Pisot number, then normalization in base $-\beta$
 on any alphabet $C$ is
realizable by a finite transducer.
\end{theorem}
\begin{proof}
The normalization is obtained by keeping in $\mathcal R_\mb(c)$ only the outputs
$y$ that are $(\mb$)-admissible. By Theorem~\ref{pis-sof} the set of admissible words is recognizable
by a finite automaton $\mathcal D_\mb$. The finite transducer 
$\mathcal N_\mb(c)$ doing the normalization is obtained by making
the intersection of the output automaton of $\mathcal R_\mb(c)$ 
with $\mathcal D_\mb$.
\end{proof}

\begin{proposition}\label{conversion}
If $\beta$ is a Pisot number, then the conversion from base $-\beta$ to base $\beta$ is realizable by a finite transducer.
The result is $\beta$-admissible.
\end{proposition}
\begin{proof}
Let $x \in I_\mb$, $x \ge 0$, such that
$\Edbm(x)=x_1x_2x_3\cdots$. 
Denote $\bar a$ the signit digit $(-a)$.
Then $\overline{x_1}x_2\overline{x_3}\cdots$ is a 
$\beta$-representation of $x$ on the alphabet 
$\widetilde{A_\mb}=\{-\lfloor \beta \rfloor , \ldots,\lfloor \beta \rfloor\}$.
Thus the conversion is equivalent to the normalization in base $\beta$ on the alphabet
$\widetilde{A_\mb}$, and when $\beta$ is a Pisot number, it is realizable by
a finite transducer by \cite{Frougny92}.
\end{proof}

\section{On-line conversion from positive to negative base}\label{s_conversion}

Proposition \ref{conversion} shows the actability of the conversion from positive to negative base with a finite transducer for a particular class of bases, {\it i.e.} the Pisot numbers. 
The result is admissible, but this transducer is not sequential.

In the case where the base is a negative integer, we have seen 
in Section~\ref{int} that the conversion from base $b$ to base $-b$
is realizable by a finite right sequential transducer.

\subsection{On-line conversion in the general case}\label{algo}

An on-line algorithm is such that, after a certain delay of latency
$\delta$ during which the data are read without writing, a digit of the output
is produced for each digit of the input, see \cite{M} for 
on-line arithmetic in integer base. 

\begin{theorem}
There exists a conversion from base $\beta$ to base $-\beta$ which is computable by an on-line
algorithm with delay $\delta$, where $\delta$ is the smallest positive integer such that
\begin{equation}\label{delay}
\frac{\lfloor \beta \rfloor}{\beta^{\delta -1}}  + \frac{\lfloor \beta \rfloor}{\beta^{\delta }}
\le 1 - \{\beta\}.
\end{equation}
The result is not admissible.
\end{theorem}

\vskip0.2cm

\hrule 
 
\vskip0.2cm
 
\noindent {\bf On-line algorithm}
\vskip0.2cm 

\hrule

\vskip0.2cm 

\noindent{\sl Input}:  a word $ (x_j)_{j \ge 1}$ of $A_\beta^\N$ such that
$x= \sum_{j \ge 1} x_j \beta^{-j}$ and $0 \le x <\frac{1}{\beta+1}$.\\
{\sl Output}: a word $ (y_j)_{j \ge 1}$ of $A_\beta^\N$
such that 
$x=\sum_{j\ge 1} y_j (-\beta)^{-j}.$\\

\noindent\texttt{begin}\\
\noindent$q_0:=0$\\
\noindent\texttt{for $j:= 1$ to $\delta$ do}\\
\hspace*{0.5cm} $q_j:=q_{j-1}+\frac{x_{j}}{\beta^j}$\\
$j:=1$\\
\noindent\texttt{while $j\ge 1$ do}\\
\hspace*{0.5cm} $z_{\delta+j} := -\beta q_{\delta+j-1}+ (-1)^j\frac{x_{\delta+j}}{\beta^\delta}$\\
\hspace*{0.5cm} \texttt{if}  $-\frac{ \beta}{\beta+1} \le z_{\delta+j} \le 
\frac{\beta^2}{\beta+1}$ \texttt{then} $y_j:= \lfloor z_{\delta+j} + \frac{\beta}{\beta+1} \rfloor$\\
\hspace*{0.5cm} \texttt{if}  $z_{\delta+j} >\frac{\beta^2}{\beta+1}$ \texttt{then} $y_j:=\lfloor \beta \rfloor$ \\
\hspace*{0.5cm} \texttt{if}  $z_{\delta+j} <-\frac{ \beta}{\beta+1}$ \texttt{then} $y_j:=0$ \\
\hspace*{0.5cm} $q_{\delta+j}:=z_{\delta+j}-y_j$\\
\hspace*{0.5cm} $j:=j+1$\\
\noindent\texttt{end}\\

\hrule
\vskip0.2cm

\begin{proof}\textbf{Claim 1.} For each $j \ge 1$
$$\frac{x_1}{\beta}+ \frac{x_2}{\beta^2}+\cdots + \frac{x_{\delta+j}}{\beta^{\delta+j}}=
-\frac{y_1}{\beta}+ \frac{y_2}{\beta^2}-\cdots + (-1)^j\frac{y_j}{\beta^{j}}+(-1)^j\frac{q_{\delta+j}}{\beta^{j}}.$$
\noindent\textbf{Claim 2.} If $-\frac{ \beta}{\beta+1} \le z_{\delta+j} \le 
\frac{\beta^2}{\beta+1}$ then $y_j$ belongs to $A_\beta$ and 
$ q_{\delta+j}$ belongs to $I_{\mb}=[-\frac{ \beta}{\beta+1} ,\frac{1}{\beta+1})$.\\
Proof of Claim 2: Clearly $0 \le y_j \le \frac{\beta^2}{\beta+1} + \frac{\beta}{\beta+1}=\beta$. Moreover, $ z_{\delta+j} + \frac{\beta}{\beta+1}=y_j + \{ z_{\delta+j} + \frac{\beta}{\beta+1}\}$, thus
$q_{\delta+j}:=z_{\delta+j}-y_j=\{ z_{\delta+j} + \frac{\beta}{\beta+1}\}  -  \frac{\beta}{\beta+1}$, and the claim is proved.\\
\noindent\textbf{Claim 3.} If 
$z_{\delta+j} >\frac{\beta^2}{\beta+1}$ then $ q_{\delta+j}>-\frac{ \beta}{\beta+1}$.\\
Proof of Claim 3: We have that $q_{\delta+j}=z_{\delta+j}-\lfloor \beta \rfloor>
\frac{\beta^2}{\beta+1} -\lfloor \beta \rfloor>-\frac{ \beta}{\beta+1}$.\\
\noindent\textbf{Claim 4.} If 
$z_{\delta+j} >\frac{\beta^2}{\beta+1}$ and
$q_{\delta+j-1} \ge -\frac{ \beta}{\beta+1}$ then $ q_{\delta+j}<\frac{1}{\beta+1}$.\\
Proof of Claim 4: Since $q_{\delta+j}= -\beta q_{\delta+j-1}+ (-1)^j\frac{x_{\delta+j}}{\beta^\delta} -\lfloor \beta \rfloor \le \frac{\beta^2}{\beta+1} + \frac{\lfloor \beta \rfloor}{\beta^\delta} - \lfloor \beta \rfloor$, the claim is proved if, and only if,
$\frac{\lfloor \beta \rfloor}{\beta^\delta} - \lfloor \beta \rfloor <1 -\beta$, that is
to say, if, and only if, $\frac{\lfloor \beta \rfloor}{\beta^\delta}<1-\{\beta\}$,
which is true thanks to~(\ref{delay}).\\
\noindent\textbf{Claim 5.} If $z_{\delta+j} <-\frac{ \beta}{\beta+1}$ and $q_{\delta+j-1}
\in I_{\mb}$ then $j$ is odd, $-  \frac{\beta}{\beta+1}- \frac{\lfloor \beta \rfloor}{\beta^\delta} \le q_{\delta+j}<-\frac{\beta}{\beta+1}$, and $ q_{\delta+j+1}$ belongs to $I_{\mb}$.\\
Proof of Claim 5: If $j$ is even then $z_{\delta+j} := -\beta q_{\delta+j-1}+ \frac{x_{\delta+j}}{\beta^\delta} > -\frac{ \beta}{\beta+1}+ \frac{x_{\delta+j}}{\beta^\delta}
\ge  -\frac{ \beta}{\beta+1}$, hence $j$ must be odd. Set $j=2k+1$.
We have $y_{2k+1}=0$ and $q_{\delta+2k+1}=z_{\delta+2k+1}=-\beta q_{\delta+2k} - \frac{x_{\delta+2k+1}}{\beta^\delta}\ge -  \frac{\beta}{\beta+1}- \frac{\lfloor \beta \rfloor}{\beta^\delta}$ since 
$q_{\delta+j-1}
\in I_{\mb}$.

Then $z_{\delta+2k+2}=-\beta q_{\delta+2k+1} + \frac{x_{\delta+2k+2}}{\beta^\delta} 
 >\frac{\beta^2}{\beta+1}$. Hence $y_{2k+2}=\lfloor \beta \rfloor$.
 By Claim 3, $q_{\delta+2k+2} >-\frac{ \beta}{\beta+1}$.\\
 On the other hand $q_{\delta+2k+2}=z_{\delta+2k+2} - \lfloor \beta \rfloor=
 -\beta q_{\delta+2k+1} + \frac{x_{\delta+2k+2}}{\beta^\delta} - \lfloor \beta \rfloor=
 \beta^2q_{\delta+2k}+\frac{x_{\delta+2k+1}}{\beta^{\delta-1}}+
 \frac{x_{\delta+2k+2}}{\beta^\delta} - \lfloor \beta \rfloor<\frac{\beta^2}{\beta+1} +
\frac{\lfloor \beta \rfloor}{\beta^{\delta-1}}+
 \frac{\lfloor \beta \rfloor}{\beta^\delta} - \lfloor \beta \rfloor \le \frac{1}{\beta+1}$
by~(\ref{delay}), thus $q_{\delta+2k+2}$ belongs to $I_\mb$.\\
 
 By hypothesis, $q_\delta$ is in $I_\mb$. By the previous claims,
 for every $k \ge 0$, $q_{\delta+2k}$ belongs to $I_\mb$ and
 $-  \frac{\beta}{\beta+1}- \frac{\lfloor \beta \rfloor}{\beta^\delta} \le q_{\delta+2k+1}<\frac{1}{\beta+1}$. Thus, for every $j \ge 1$, 
 $$\frac{x_1}{\beta}+ \cdots + \frac{x_{\delta+j}}{\beta^{\delta+j}}=
\frac{y_1}{(-\beta)}+ \cdots + \frac{y_{j}}{(-\beta)^{j}}+\frac{q_{\delta+j}}{(-\beta)^{j}}$$
with $q_{\delta+j}$ bounded. Therefore the algorithm converges, and 
$$\sum_{j \ge 1} x_j \beta^{-j}=\sum_{j\ge 1} y_j (-\beta)^{-j}.$$
\end{proof}

\subsection{Conversion in the Pisot case}

We now show that, when $\beta$ is a Pisot number, there is a finite on-line transducer
realizing the conversion.

\begin{theorem}
If $\beta$ is a Pisot number, the conversion from base $\beta$ to base $\mb$ is realizable
by a finite on-line transducer.
\end{theorem}

\begin{proof}
Following the on-line algorithm of Section~\ref{algo}
we construct an on-line transducer $\mathcal{C}$ as follows. The set of states is $Q=Q_t \cup Q_s$, with the set of transient states $Q_t=\{q_j \mid 0 \le j \le \delta-1\}$,
and the set of synchronous states $Q_s=\{q_{\delta+j} \mid j \ge 0\}$.
The initial state is $q_0$.
For $1 \le j \le
\delta$,
transient edges are defined by
$$q_{j-1} \stackrel{x_j|\varepsilon}{\longrightarrow}q_j.$$

Synchronous edges are defined by
$$q_{\delta+j-1}
\stackrel{x_{\delta+j}|y_{j}}{\longrightarrow}q_{\delta+j}$$
for $j \ge 1$.
There
 is an infinite path in the automaton $\mathcal C$ starting in $q_0$
and labelled by
$$q_0\stackrel{x_1|\varepsilon}{\longrightarrow}q_1 \cdots
\stackrel{x_\delta|\varepsilon}{\longrightarrow}q_\delta
\stackrel{x_{\delta+1}|y_1}{\longrightarrow}q_{\delta+1}
\stackrel{x_{\delta+2}|y_2}{\longrightarrow}q_{\delta+2}\cdots $$
if, and only if, $\sum_{j \ge 1} x_j
\beta^{-j}=\sum_{j \ge 1} y_j (-\beta)^{-j}$.

\bigskip

Let $M_\beta(X)$ be the minimal polynomial of $\beta$ and let $\beta=\beta_1, \beta_2,\ldots,\beta_d$ be the roots of $M_\beta$. Recall that $\Z[X]/(M_\beta(X)) \sim
\Z[\beta] $ is a discrete lattice of rank $d$.
Since $\beta$ is a Pisot number, $|\beta_i|<1$
for $2 \le i \le d$. 

For each $j \ge 1$, $q_{j}$ is an element
of $\Z[\beta, \beta^{-1}]$. For $1 \le i \le d$ let $q_{j}(\beta_i)$ be the element of
$\Z[\beta_i, \beta_i^{-1}]$ obtained by replacing $\beta$ by $\beta_i$
in $q_j$. Then $q_{j}=q_{j}(\beta)$.

First of all, for every $j \ge 1$, $-  \frac{\beta}{\beta+1}- \frac{\lfloor \beta \rfloor}{\beta^\delta} \le q_{j}(\beta)<\frac{1}{\beta+1}$ by the on-line algorithm.

Secondly, for every $j \ge 1$ and $2 \le i \le d$,
\begin{equation}\label{bou}
q_{\delta+j}(\beta_i)=-\beta_i q_{\delta+j-1}(\beta_i)+(-1)^j\frac{x_{\delta+j}}{\beta_i^{\delta}}
-y_{j}.
\end{equation}
For $2 \le i \le d$ let
$$M_i=\frac{\lfloor \beta \rfloor}{(1-|\beta_i|)}\big(1
+\frac{1}{|\beta_i|^{\delta}}\big).$$
Then, if $|q_{\delta+j-1}(\beta_i)| \le M_i$, then $|q_{\delta+j}(\beta_i)| \le M_i$
by (\ref{bou}).

Now, for $0 \le j \le \delta$ and $2 \le i \le d$, 
$$|q_j(\beta_i)| < \lfloor \beta \rfloor (\frac{1}{|\beta_i|}
+ \cdots + \frac{1}{|\beta_i|^{\delta}})<M_i.$$
Define a norm on $\Z[X]/(M_\beta(X))$ by
$\Vert q \Vert= \max_{1 \le i \le d} |q(\beta_i)|$.
Thus the elements of $Q$ are all bounded in norm, and so $Q$ is finite.
\end{proof}

\bigskip

In the particular case that $\beta^2=a \beta +1$ ($\beta$ is thus a Pisot number)
we can construct directly a simpler finite left sequential transducer realizing the conversion.

\begin{proposition}
If 
$\beta^2=a \beta +1$, $a \ge 1$, then the conversion from base $\beta$ to base $\mb$
is realizable by the finite left sequential transducer of Fig.~\ref{quad}. 
\end{proposition}
\begin{proof}
The left sequential transducer in Fig.~\ref{quad}
converts a $\beta$-expansion of a real number $x$ in $[0,\beta)$
of the form
$x_0 \decdot x_1 x_2 \cdots$ into a $(\mb)$-representation of
$x$ of the form $y_0 \decdot y_1 y_2 \cdots$. 
We take $0 \le d \le a$, $0 \le c \le a-1$, $1 \le e \le a$.
Since the input is admissible, no factor $ae$, with $1 \le e \le a$ can occur.

 \begin{figure}[h]
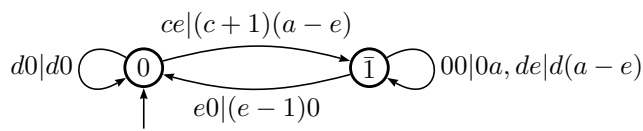

 \begin{center}
 \VCDraw{%
 \begin{VCPicture}{(-1,-1)(4,2)}
 % states
 \State[0]{(-1,0)}{A}
 \State[\bar 1]{(4,0)}{B}
 \Initial[s]{A} 
 % transitions
 \ArcL[.5]{B}{A}{e0|(e-1)0}
 \ArcL[.5]{A}{B}{ce|(c+1)(a-e)}
 \LoopW[.5]{A}{d0|d0}
 \LoopE[.5]{B}{00|0a,de|d(a-e)}

 \end{VCPicture}%
         }
 \end{center}
 \caption{Finite left sequential transducer realizing conversion from base $\beta$ to base $\mb$, $\beta^2=a \beta +1$}\label{quad}
 \end{figure}

\end{proof}

%%%%

\end{document}